\documentclass[sigconf]{acmart}

\newif\iflong
\longtrue %

\usepackage[vlined,linesnumbered,noresetcount]{algorithm2e}

\usepackage{caption}
\usepackage{color}
\usepackage{enumitem}
\usepackage{setspace}
\usepackage{tikz}
\usepackage{titlesec}
\usetikzlibrary{decorations.pathreplacing,positioning,automata,calc}
\usetikzlibrary{shapes,arrows}
\usepgflibrary{shapes.symbols}
\usetikzlibrary{shapes.symbols}
\usetikzlibrary{patterns}
\usetikzlibrary{decorations.pathreplacing}
\usetikzlibrary{decorations.markings}

\setlist[enumerate]{leftmargin=1.5em,itemsep=-2pt,topsep=2pt}

\newcommand{\ag}[1]{}
\newcommand{\ps}[1]{}
\newcommand{\fr}[1]{}

\newcommand{\propose}{\mathtt{propose}}
\newcommand{\decide}{\mathtt{decide}}
\newcommand{\starttimer}{\mathtt{start\_timer}}

\newcommand{\MPropose}{\mathtt{Propose}}
\newcommand{\MDecide}{\mathtt{Decide}}
\newcommand{\MoneA}{\mathtt{1A}}
\newcommand{\MtwoA}{\mathtt{2A}}
\newcommand{\MoneB}{\mathtt{1B}}
\newcommand{\MtwoB}{\mathtt{2B}}

\newcommand{\val}{\mathsf{val}}
\newcommand{\initial}{\mathsf{initial\_val}}
\newcommand{\decided}{\mathsf{decided}}
\newcommand{\bal}{\mathsf{bal}}
\newcommand{\vbal}{\mathsf{vbal}}
\newcommand{\proposer}{\mathsf{proposer}}
\newcommand{\timer}{\mathsf{new\_ballot\_timer}}

\newcommand{\bmax}{b_{\rm max}}

\newcommand{\lval}{\mathit{val}}
\newcommand{\lvbal}{\mathit{vbal}}
\newcommand{\ldecided}{\mathit{decided}}
\newcommand{\lproposer}{\mathit{proposer}}

\SetKwBlock{SubAlgoBlock}{}{end}
\newcommand{\SubAlgo}[2]{#1 \SetAlgoVlined\SubAlgoBlock{\SetAlgoNoLine\SetAlgoNoEnd\renewcommand{\;}{\\} #2}}

\newcommand{\To}{\textbf{to}\xspace}
\newcommand{\Let}{\textbf{let}\xspace}
\newcommand{\send}{\textbf{send}\xspace}
\newcommand{\from}{\textbf{from}\xspace}

\newcommand{\assign}[2]{\ensuremath{#1} \ensuremath{\leftarrow} \ensuremath{#2}}
\newcommand{\fromall}{\textbf{from all}\xspace}
\newcommand{\precond}{\textbf{pre:}\xspace}
\newcommand{\onreceive}{\textbf{when received}\xspace}

\newtheoremstyle{shrdef}%
   {.5\baselineskip\@plus.2\baselineskip
     \@minus.2\baselineskip}%
   {.5\baselineskip\@plus.2\baselineskip
     \@minus.2\baselineskip}%
   {\itshape}%
   {}%
   {\sc}%
   {.}%
   {.5em}%
   {\thmname{#1}\thmnumber{ #2}\thmnote{ {\sc(#3)}}}%
\theoremstyle{shrdef}
\newtheorem{definition}{Definition}

\newtheorem{theorem}{Theorem}
\newtheorem{lemma}{Lemma}

\newcommand{\cardinalOf}[1]{\ensuremath{\lvert #1 \rvert}}
\newcommand{\inter}{\ensuremath{\cap}}
\newcommand{\union}{\ensuremath{\cup}}
\newcommand{\procSet}{\ensuremath{\Pi}}
\newcommand{\run}{\ensuremath{\sigma}}
\newcommand{\steps}[3]{\ensuremath{{{\color{black}{\mathbf{#1}}}}^{#2}}_{#3}}
\newcommand{\crash}{\ensuremath{\mathtt{crash}}}

\tikzstyle{message} = [draw, -latex',black!60, shorten <=0.2em]
\tikzstyle{dmessage} = [draw, -latex',black!60, shorten >=1em, dashed]
\tikzstyle{messageA} = [draw, -latex',blue!40, shorten <=0.2em]
\tikzstyle{messageB} = [draw, -latex',red!40, shorten <=0.2em]
\tikzstyle{lifeline} = [draw, gray!40]

\renewcommand{\frac}[2]{\ensuremath{#1/#2}}

\makeatletter
\newcommand{\removelatexerror}{\let\@latex@error\@gobble}
\makeatother

\newcommand{\tr}[2]{\iflong{}\S#1\else{}\cite[\S#2]{ext}\fi}
\newcommand{\tra}[2]{\iflong{}(\S#1)\else{}\cite[\S#2]{ext}\fi}

\begin{document}

\title[Revisiting Lower Bounds for Two-Step Consensus]
{Revisiting Lower Bounds for Two-Step Consensus\\\!}

\author{Fedor Ryabinin}
\affiliation{%
  \institution{IMDEA Software Institute\\Universidad Politécnica de Madrid}
  \city{Madrid}
  \country{Spain}}

\author{Alexey Gotsman}
\affiliation{%
  \institution{IMDEA Software Institute}
  \city{Madrid}
  \country{Spain}}

\author{Pierre Sutra}
\affiliation{%
  \institution{Télécom SudParis, Inria Saclay\\Institut Polytechnique de Paris}
  \city{Palaiseau}
  \country{France}}

\begin{abstract}
  A seminal result by Lamport shows that at least $\max\{2e+f+1, 2f+1\}$ processes are required to implement partially synchronous consensus that tolerates $f$ process failures and can furthermore decide in two message delays under $e$ failures.
  This lower bound is matched by the classical Fast Paxos protocol.
  However, more recent practical protocols, such as Egalitarian Paxos, provide two-step decisions with fewer processes, seemingly contradicting the lower bound.
  We show that this discrepancy arises because the classical bound requires two-step decisions under a wide range of scenarios, not all of which are relevant in practice.
  We propose a more pragmatic condition for which we establish tight bounds on the number of processes required.
  Interestingly, these bounds depend on whether consensus is implemented as an atomic object or a decision task.
  For consensus as an object, $\max\{2e+f-1, 2f+1\}$ processes are necessary and sufficient for two-step decisions, while for a task the tight bound is $\max\{2e+f, 2f+1\}$.
\end{abstract}




\copyrightyear{2025}
\acmYear{2025}
\setcopyright{acmlicensed}\acmConference[PODC '25]{ACM Symposium on Principles of Distributed Computing}{June 16--20, 2025}{Huatulco, Mexico}
\acmBooktitle{ACM Symposium on Principles of Distributed Computing (PODC '25), June 16--20, 2025, Huatulco, Mexico}
\acmDOI{10.1145/3732772.3733541}
\acmISBN{979-8-4007-1885-4/2025/06}

\maketitle

\bigskip

\section{Introduction}
\label{sec:intro}

\paragraph{Context.}
Consensus is a fundamental abstraction in distributed computing, widely used in practice for state-machine replication.
It is well-known that partially synchronous consensus tolerating $f$ process crashes requires at least $2f + 1$ processes \cite{dls}.
Protocols that match this bound, such as Paxos, are usually leader-driven.
If the system is synchronous and the initial leader process is correct, these protocols can decide within two message delays; but if the initial leader is faulty, the latency increases.
This issue is addressed by {\em fast} consensus protocols~\cite{fast-paxos,gbroadcast} that can decide in two message delays under {\em any} process failures up to a given threshold $e \le f$.
This is usually done via a {\em fast path}, which avoids the leader.
Lamport showed that fast consensus requires at least $\max\{2e + f + 1, 2f+1\}$ processes~\cite{lower-bounds}, a bound matched by the Fast Paxos protocol~\cite{fast-paxos}.
But surprisingly, a more recent protocol called Egalitarian Paxos~\cite{epaxos} decides within two message delays under $e=\lceil\frac{(f+1)}{2}\rceil$ failures while using only $2f+1 = 2e + f - 1$ processes.
Lamport's bound in this case requires $2f+3 = 2e + f + 1$ processes. What's going on?

\paragraph{Contribution.}
To resolve this conundrum, we revisit Lamport's definition of fast protocols.
Roughly speaking, the existing definition considers a consensus protocol {\em fast} if for every proposer $p$ and {\em every} correct process $q$, there exists a run in which $p$ is the only process that sends its proposal, and $q$ decides in two message delays.
However, in practical protocols a client typically submits its proposal to one of the processes participating in consensus -- a {\em proxy} -- which replies to the client with the decision \cite{smr}.
In this setting, we would like the proxy to decide fast, but the speed of the decision at other processes is irrelevant.
We thus propose a more pragmatic definition that reflects this consideration.
We consider a protocol fast if: {\em (i)} for all initial configurations, {\em some} process can decide within two message delays; and {\em (ii)} starting from configurations where all correct processes propose the same value, {\em every} correct process can decide in two message delays.
We thus require fast decisions at all processes only when the proposals are the same; otherwise, only a single process is required to decide fast.
We show that depending on formulation of consensus, satisfying our definition requires up to two fewer processes than Lamport's.
This difference is practically significant for wide-area deployments, where contacting an additional process may incur a cost of hundreds of milliseconds per command.

The standard definition of consensus is that of a \emph{decision task}: each process has an initial value, and processes try to agree on one of these.
We first show that implementing a fast consensus task is possible if and only if the number of processes is at least $\max\{2e + f, 2f+1\}$.
An alternative definition of consensus is that of an {\em atomic
  object}~\cite{unifying}, where processes explicitly specify their proposal $v$
by calling an operation $\propose(v)$; in particular, this allows a process not to make any proposal at all.
The consensus is required to be linearizable and wait-free.
We show that this version of fast consensus can be solved if and only if the number of processes is at least $\max\{2e + f - 1, 2f+1\}$.

\paragraph{Related work.}
In synchronous systems, bounds on consensus latency were investigated by Charron-Bost and Schiper~\cite{charron04}.
Other work considered bounds on consensus in asynchronous systems with failure detectors~\cite{omega}, showing that two-step decisions are possible if the failure detector selects a correct initial leader~\cite{zero-degradation}.
In this paper we do not rely on this assumption and require fast decisions without any additional information.
The work in \cite{leaderless-disc20} shows that $2e + f - 1$ is a lower bound for a limited class of consensus object protocols constructed similarly to Egalitarian Paxos.
In comparison to that work, our lower bounds apply to arbitrary consensus protocols, both task and object, and we provide matching upper bounds.
Kuznetsov et al.~\cite{bft-fast-lower-bound} studied the relation between the thresholds $f$ and $e$ under Byzantine failures, and proved that $3f+2e-1$ processes are necessary and sufficient for a protocol to be fast according to a definition close to Lamport's. 
An interesting direction of future work is to combine their techniques with ours to see if this can lower the number of processes required in the Byzantine case.

\section{Results}
\label{sec:results}

We assume a system $\procSet = \{p_1, \dots, p_n\}$ of $n \geq 3$ crash-prone processes communicating via reliable links.
The system is partially synchronous~\cite{dls}: after some global stabilization time ($\mathsf{GST}$) messages take at most $\Delta$ units of time to reach their destinations.
The bound $\Delta$ is known to the processes, while $\mathsf{GST}$ is unknown.
We say that events that happen during the time interval $[0, \Delta)$ form {\em the first round}, events that happen during the time interval $[\Delta, 2\Delta)$ {\em the second round}, and so on.

In the {\em consensus} decision task, each process has an input value (a {\em proposal}) and may output a value (a {\em decision}).
It is required that:
every decision is the proposal of some process \emph{(Validity)};
no two decisions are different \emph{(Agreement)}; and
every correct process eventually decides \emph{(Termination)}.
Consensus can alternatively be modeled as a shared object, where processes explicitly propose a value $v$ by calling $\propose(v)$, which eventually returns the decision.
This allows a process not to propose anything at all.

\begin{definition}
  A protocol is \textbf{$\boldsymbol{f}$-resilient} if it achieves consensus under at most $f$ processes failures.
\end{definition}

While an $f$-resilient protocol can guarantee termination with up to $f$ process failures, these failures can nevertheless make the protocol run slow even when the system is synchronous.
We thus also consider another threshold, $e$, which determines the maximal number of failures that a protocol can tolerate while still providing fast, two-step decisions in synchronous runs (we assume $e\le f$).
We next define this kind of runs formally.
\begin{definition}
  Given $E \subseteq \procSet$, a run is \textbf{$\boldsymbol{E}$-faulty synchronous}, if:
  \begin{enumerate}[noitemsep,topsep=0pt,parsep=0pt]
  \item All processes in $\procSet \setminus E$ are correct, and all processes in $E$ are faulty.
  \item Processes in $E$ crash at the beginning of the first round.
  \item All messages sent during a round are delivered precisely at the beginning of the next round.
  \item All local computations are instantaneous.
  \end{enumerate}
\end{definition}

We next define protocols that can provide two-step decisions in synchronous runs with up to $e$ failures.
We formulate our definitions for tasks and defer their analogs for objects to \tr{\ref{section:object-def}}{A}.
\begin{definition}
A run is \textbf{two-step} for a process $p$ if in this run $p$ decides a value by time $2\Delta$.
\end{definition}
\begin{definition}
  \label{def:fast-task}
  A protocol 
  is \textbf{$\boldsymbol{e}$-two-step} if for all $E \subseteq \procSet$ of size $e$:
  \begin{enumerate}[noitemsep,topsep=0pt,parsep=0pt]
  \item     \label{two-step-one}
    For every initial configuration $I$, there exists an $E$-faulty synchronous run starting from $I$ which is two-step for some process.
  \item \label{two-step-two}
    For every initial configuration $I$ in which processes in $\procSet \setminus E$ propose the same value, for each process $p \in \procSet \setminus E$, there exists an $E$-faulty synchronous run starting from $I$ which is two-step for $p$.
    \end{enumerate}
\end{definition}
Differently from fast consensus protocols introduced by Lamport~\cite{lower-bounds}, we require only a single process to decide fast (item~\ref{two-step-one}), unless all the proposals are the same (item~\ref{two-step-two}).
Paxos is not $e$-two-step for any $e>0$, while Fast Paxos~\cite{fast-paxos} is $e$-two-step if
$n \ge \max\{2e + f + 1, 2f+1\}$.
We are interested in the following problem: \emph{what is the minimal number of processes to solve $f$-resilient $e$-two-step consensus?}
Surprisingly, this bound is lower than for Fast Paxos and furthermore depends on the problem definition.
\begin{theorem}
  \label{lower:task}
  An $f$-resilient $e$-two-step consensus task is implementable iff $n \ge \max\{2e + f, 2f+1\}$.
\end{theorem}
\begin{theorem}
  \label{lower:object}
  An $f$-resilient $e$-two-step consensus object is implementable iff $n \ge \max\{2e + f - 1, 2f+1\}$.
\end{theorem}
We defer the proofs of the lower bounds to \tr{\ref{section:lower-proof}}{B}, and present matching upper bounds next.

\section{Upper Bounds}
\label{sec:upper}

\begin{figure*}[t]
  \removelatexerror
  \begin{minipage}{0.91\linewidth}
  \scalebox{0.87}{
	\begin{minipage}{0.55\linewidth}
	  \begin{algorithm*}[H]
        \setstretch{0.9}
		\renewcommand{\;}{\\}
        \DontPrintSemicolon

        \SubAlgo{\bf at startup {\color{purple}/ upon an invocation of $\propose(v)$}}{
          \label{alg:startup}
          {\color{purple}
            \If{$\val = \bot$}{
              \label{alg:startup-precond}
                $\assign{\initial}{v}$\;
                \label{alg:set-init-val}
              {\color{black}
                \send $\MPropose(\initial)$ \To{} $\Pi \setminus \{p_i\}$\;
                \label{alg:send-propose}
              }
            }
          }
        }

        \smallskip

        \SubAlgo{\onreceive $\MPropose(v)$ \from $q$}{
          \precond $\bal = 0 \land \val = \bot \land v \geq \initial$ $~~~~~~\,{\color{purple}{}\land (\initial \neq \bot \Longrightarrow v = \initial)}$\;
          \label{alg:propose-precond}
          $\assign{(\val, \proposer)}{(v, q)}$\;
          \send $\MtwoB(0, v)$ \To $q$\;
        }

        \smallskip

        \SubAlgo{\onreceive $\MtwoB(\bal, v)$ \fromall $q \in P$}{
          \precond $(\bal = 0 \land |P \cup \{p_i\}| \geq n - e \land \val \in \{\bot, v\}) \lor{}$
          \makebox[14pt]{}$(\bal \neq 0 \land |P| \geq n - f)$\;
          \label{alg:two-b-precond}
          $\assign{(\val, \decided)}{(v, v)}$\;
          {\bf decide} $v$\;
          \label{alg:learn-coord}
          \send $\MDecide(v)$ \To{} $\Pi \setminus \{p_i\}$\;
          \label{alg:send-learn}
        }

        \smallskip

        \SubAlgo{\onreceive $\MDecide(v)$}{
          $\assign{(\val, \decided)}{(v, v)}$\;
          \label{alg:learn-set-v}
          {\bf decide} $v$\;
        }

        \smallskip

        \SubAlgo{\onreceive $\MoneA(b)$ \from $q$\label{alg:receive1a}}{
          \precond $b > \bal$\;
          \label{alg:one-a-precond}
          $\assign{\bal}{b}$\;
          \label{alg:one-a-set-bal}
          \send $\MoneB(b, \vbal, \val, \proposer, \decided)$ \To $q$\;
          \label{alg:send-one-b}
        }

        \smallskip
        \smallskip
        \smallskip
        
	  \end{algorithm*}
	\end{minipage}
    \begin{minipage}{0.6\linewidth}
      \begin{algorithm}[H]
        \setstretch{0.9}
        \renewcommand{\;}{\\}
        \DontPrintSemicolon

        \SubAlgo{{\bf on timeout}}{
          \label{alg:timeout}
          \Let $b = (\mbox{\rm a ballot $> \bal$ such that}$ $i \equiv b\ ({\tt mod}\ n))$\;
          \label{alg:new-ballot}
          \send $\MoneA(b)$ \To{} $\Pi$\;
          \label{alg:send-one-a}
        }

        \smallskip

        \SubAlgo{\onreceive $\MoneB(b, \lvbal_q, v_q, \lproposer_q, \ldecided_q)$ \fromall{} $q \in Q$\label{alg:receive1b}}{
          \precond $|Q| = n-f$\;
          \label{alg:one-b-cond}
          \Let $\lval = \bot$\;
          \Let $\bmax = {\tt max}\{\lvbal_q \mid q \in Q\}$\;
          \label{alg:let-bmax}
          \Let $R= \{q \in Q \mid \lproposer_q \notin Q\}$\;
          \label{alg:q-prime}
          \uIf{$\exists q \in Q.\, \ldecided_q \neq \bot$}{
            \label{alg:cond-learn}
            $\assign{\lval}{\ldecided_q}$\;
            \label{alg:set-learn}
          }
          \uElseIf{$\bmax > 0$}{
            \label{alg:cond-slow-ballot}
            $\assign{\lval}{v_q \text{\ such that } \lvbal_q = \bmax}$\;
            \label{alg:set-slow-ballot}
          }
          \uElseIf{$\exists v \neq \bot.\,\exists S \subseteq R.\, |S| > n-f-e \land \forall q \in S.\, v_q = v$}{
            \label{alg:cond-1}
            $\assign{\lval}{v}$\;
            \label{alg:set-lval-cond-1}
          }
          \uElseIf{$\exists v \neq \bot.\, \exists S \subseteq R.\, |S| = n-f-e \land \forall q \in S.\, v_q = v$}{
            \label{alg:cond-2}
            $\assign{\lval}{\mbox{the maximal value $v$ satisfying the condition at line~\ref{alg:cond-2}}}$\;
            \label{alg:set-lval}
          }
          \ElseIf{$\initial \neq \bot$}{
            $\assign{\lval}{\initial}$\;
            \label{alg:set-val-initial}
          }
          \lIf{$\lval \neq \bot$}{\send $\MtwoA(b, \lval)$ \To{} $\Pi$}
          \label{alg:send-two-a}
        }

        \smallskip

        \SubAlgo{\onreceive $\MtwoA(b, v)$ \from $q$}{
          \precond $\bal \leq b$\;
          \label{alg:two-a-cond}
          $\assign{(\val, \bal, \vbal)}{(v, b, b)}$\;
          \label{alg:two-a-set-vars}
          \send $\MtwoB(b, v)$ \To{} $q$\;
          \label{alg:send-two-b}
        }
      \end{algorithm}
	\end{minipage}
  }
\end{minipage}
\vspace{-5pt}
  \caption{Consensus task at a process $p_i$.
    {\color{purple} Red lines} highlight the changes needed to implement a
    consensus object.
  }
  \label{fig:alg}
\end{figure*}

Figure~\ref{fig:alg} presents a consensus protocol that matches our lower bounds.
With the {\color{purple} red lines} ignored, the protocol implements a consensus task for $\max\{2e+f, 2f+1\}$ processes (Theorem~\ref{lower:task}); with the {\color{purple} red lines} included, it implements a consensus object for $\max\{2e+f-1, 2f+1\}$ processes (Theorem~\ref{lower:object}).
We first describe the task version, and then briefly highlight the differences for the object version.
Our protocol improves the classical Fast Paxos~\cite{fast-paxos} to reduce the number of processes required.
It operates in a series of {\em ballots}, with each process storing the current ballot in a variable $\bal$.
The initial ballot $0$ is called the {\em fast ballot}, while all others are {\em slow ballots}.
During a ballot processes attempt to reach an agreement on one of the initial proposals by exchanging their votes, tracked in a variable $\val$. %
Processes maintain a variable $\vbal$ to track the last ballot in which they cast a vote.

\paragraph{Fast ballot.}
Each process starts by sending its proposal (stored in a variable $\initial$) to all other processes in a $\MPropose$ message (line~\ref{alg:send-propose}).
A process accepts the message from a process $p$ only if it is in ballot $0$, has not yet voted, and the value received is greater than or equal to its own proposal (line~\ref{alg:propose-precond}).
In this case it updates $\val$ and replies to $p$ with a $\MtwoB$ message, analogous to the eponymous message of Paxos.
A process that gathers support from $n-e$ processes including itself (the first disjunct at line~\ref{alg:two-b-precond}) considers its value decided (line~\ref{alg:learn-coord}) and communicates this to the other processes via a $\MDecide$ message (line~\ref{alg:send-learn}).

This flow guarantees that the protocol is $e$-two-step.
Indeed, consider an initial configuration $I$, a set $E \subseteq \Pi$ of size $e$, and a process $p \not\in E$ that proposes the highest value $v$ in $I$ among all processes in $\Pi \setminus E$.
Then there exists an $E$-faulty synchronous run in which the $\MPropose$ message sent by $p$ is the first one accepted by all other correct processes.
Since $p$ does not accept any $\MPropose$ message for value lower than $v$ (line~\ref{alg:propose-precond}), it will be able to collect the $\MtwoB$ messages from $n-e-1$ other processes.
It will then satisfy the condition at line~\ref{alg:two-b-precond} and decide by $2\Delta$.

\paragraph{Slow ballots.}
If processes do not reach an agreement at ballot $0$ within $2\Delta$, the protocol nominates a process $p_i$ to initiate a new slow ballot (line~\ref{alg:new-ballot}); this nomination is done using standard techniques \tra{\ref{section:full-new-ballot}}{C}.
Process $p_i$ broadcasts a $\MoneA$ message asking the others to join the new ballot (line~\ref{alg:send-one-a}), and processes respond with a $\MoneB$ message, carrying information about their state (line~\ref{alg:receive1a}).
Once $p_i$ receives $n-f$ replies from a set $Q$ (line~\ref{alg:receive1b}), it computes a proposal for its ballot as we describe in the following.
It then sends this proposal in a $\MtwoA$ message (line~\ref{alg:send-two-a}), to which processes reply with $\MtwoB$ messages (line~\ref{alg:send-two-b}).
Process $p_i$ decides after collecting $n-f$ matching votes (the second disjunct at line~\ref{alg:two-b-precond}).

\paragraph{Computing the proposal.}
Process $p_i$ computes its proposal based on the states received in the $\MoneB$ messages.
If some process has already decided a value, then $p_i$ selects that value (line~\ref{alg:cond-learn}).
Otherwise, $p_i$ considers the highest ballot $\bmax$ where a vote was cast (line~\ref{alg:let-bmax}): the votes in $\bmax$ supersede those in lower ballots.
If $\bmax > 0$, then $p_i$ selects the associated value, just like in the usual Paxos (line~\ref{alg:cond-slow-ballot}).

If $\bmax = 0$, then some value may have been decided on the fast path.
To handle this case, process $p_i$ first excludes the values whose proposers belong to $Q$ (line~\ref{alg:q-prime}): these proposers have not taken the fast path, because line~\ref{alg:cond-learn} did not execute; they will not take it in the future either, because they have moved to a slow ballot when replying to $p_i$ (line~\ref{alg:one-a-set-bal}).
If after this there is a value $v$ with more than $n-f-e$ votes, then $p_i$ proposes $v$ (line~\ref{alg:cond-1}); we prove below that $v$ is unique.
If no such $v$ exists, $p_i$ considers the values with exactly $n-f-e$ votes (line~\ref{alg:cond-2}).
There may be multiple such values, so $p_i$ selects the greatest one (line~\ref{alg:set-lval}).
Finally, in all other cases $p_i$ selects its own initial value (line~\ref{alg:set-val-initial}).
The following lemma shows that this algorithm correctly recovers fast path decisions.

\begin{lemma}
  Assume $n \geq 2e + f$.
  If a value $v$ is decided via the fast path at ballot $0$, then lines~\ref{alg:cond-1}--\ref{alg:set-val-initial} always select $v$ as $\lval$.
  \label{lem:fast-path-rec}
\end{lemma}

{\sc Proof.}
Since $v$ takes the fast path, there are at least $n-e$ processes that vote for $v$ at ballot $0$ (implicitly including the proposer of $v$).
  The set of these processes intersects with $Q$ (line~\ref{alg:receive1b}) in $\ge n-e-f$ processes, thus enabling the condition at either line~\ref{alg:cond-1} or line~\ref{alg:cond-2}.
  
  Consider first the case of line~\ref{alg:cond-1}, so that $> n-e-f$ processes in $Q$ voted for $v$.
  Then the number of processes in $Q$ that could have voted for another value is $< (n-f)-(n-e-f) = e = (2e + f) - e-f \le n - e - f$.
  Hence, $v$ is the only value that can satisfy the condition at line~\ref{alg:cond-1}, and $\lval$ must be assigned to $v$ at line~\ref{alg:set-lval-cond-1}.
  
  Assume now that the condition at line~\ref{alg:cond-2} holds, so that $n-e-f$ processes in $Q$ voted for $v$.
  Since $v$ went onto the fast path, $\ge n-e$ processes voted for $v$ overall.
  Thus, all $f$ processes outside $Q$ must have voted for $v$.
  Since the proposer of any other value $v'$ that also satisfies the condition at line~\ref{alg:cond-2} does not belong to $Q$ (line~\ref{alg:q-prime}), it must have voted for $v$.
  But then $v' < v$ by line~\ref{alg:propose-precond}, which proves that $\lval$ is assigned to $v$ at line~\ref{alg:set-lval}.\qed

\paragraph{Consensus object.}
When implementing a consensus object, the variable $\initial$
(initially $\bot$, lower than any other value)
is explicitly updated upon an invocation of $\propose$ (line~\ref{alg:set-init-val}), provided the process has not yet voted for someone else's proposal.
The only other difference is the additional condition at line~\ref{alg:propose-precond}:
a process responds to the $\MPropose(v)$ message only if it has not proposed yet ($\initial = \bot$), or if $v$ matches its own proposal ($\initial = v$).

\smallskip

In \tr{\ref{section:upper-proof}}{C} we prove that the protocols we have just presented indeed satisfy the conditions of Theorems~\ref{lower:task} and~\ref{lower:object}.

\begin{acks}
This work was supported by projects BYZANTIUM, DECO and PRODIGY
funded by MCIN/AEI, and REDONDA funded by CHIST-ERA.
We thank Petr Kuznetsov for helpful discussions. We also thank 
Benedict Elliot Smith, whose work on fast consensus in
Cassandra served as a practical motivation for this paper.
\end{acks}

\bibliographystyle{ACM-Reference-Format}
\bibliography{main}

@misc{ext,
      title={Revisiting Lower Bounds for Two-Step Consensus},
      author={Fedor Ryabinin and Alexey Gotsman and Pierre Sutra},
      year={2025},
      eprint={0},
      archivePrefix={arXiv},
      primaryClass={cs.DC}
}

@inproceedings{bft-fast-lower-bound,
  author    = {Petr Kuznetsov and
               Andrei Tonkikh and
               Yan X Zhang},
  title     = {Revisiting Optimal Resilience of Fast {B}yzantine Consensus},
  year      = {2021},
  booktitle = {Symposium on Principles of Distributed Computing ({PODC})},
}

@article{dls,
  author  = {Cynthia Dwork and
             Nancy A. Lynch and
             Larry J. Stockmeyer},
  title   = {{Consensus in the Presence of Partial Synchrony}},
  journal = {J. {ACM}},
  volume  = {35},
  number  = {2},
  year    = {1988},
}

@inproceedings{epaxos,
  author    = {Iulian Moraru and
               David G. Andersen and
               Michael Kaminsky},
  title     = {{There Is More Consensus in Egalitarian Parliaments}},
  booktitle = {Symposium on Operating Systems Principles ({SOSP})},
  year      = {2013},
}

@article{fast-paxos,
  author  = {Leslie Lamport},
  title   = {{Fast Paxos}},
  journal = {Distributed Computing},
  year    = {2006},
  volume  = {19},
}

@inproceedings{gbroadcast,
  author    = {Fernando Pedone and
               Andr{\'{e}} Schiper},
  title     = {{Generic Broadcast}},
  booktitle = {International Symposium on Distributed Computing ({DISC})},
  year      = {1999},
}

@article{lower-bounds,
  author  = {Leslie Lamport},
  title   = {{Lower Bounds for Asynchronous Consensus}},
  journal = {Distributed Computing},
  year    = {2006},
  volume  = {19},
}

@article{omega,
  author       = {Tushar Deepak Chandra and Sam Toueg},
  title        = {Unreliable Failure Detectors for Reliable Distributed Systems},
  journal      = {J. {ACM}},
  volume       = {43},
  number       = {2},
  xpages        = {225--267},
  year         = {1996},
}

@article{smr,
  author  = {Fred B. Schneider},
  title   = {{Implementing Fault-Tolerant Services Using the State Machine Approach: A Tutorial}},
  journal = {ACM Computing Surveys},
  year    = {1990},
  volume  = {22},
}

@article{unifying,
author = {Casta\~{n}eda, Armando and Rajsbaum, Sergio and Raynal, Michel},
title = {Unifying Concurrent Objects and Distributed Tasks: Interval-Linearizability},
year = {2018},
volume = {65},
number = {6},
journal = {J. ACM},
}

@inproceedings{zero-degradation,
author    = {Partha Dutta and
             Rachid Guerraoui},
title     = {Fast Indulgent Consensus with Zero Degradation},
year      = {2002},
booktitle = {European Dependable Computing Conference ({EDCC})},
}

@inproceedings{leaderless-disc20,
  author    = {Tuanir {Fran{\c{c}}a Rezende} and
               Pierre Sutra},
  title     = {Leaderless State-Machine Replication: Specification, Properties, Limits},
  booktitle = {International Symposium on Distributed Computing ({DISC})},
  year      = {2020},
}

@article{charron04,
author = {Charron-Bost, Bernadette and Schiper, Andr\'{e}},
title = {Uniform consensus is harder than consensus},
year = {2004},
volume = {51},
number = {1},
journal = {J. Algorithms},
xpages = {15–37},
}

\iflong
  \clearpage
  \appendix
  \section{$e$-Two-Step Consensus Objects}
\label{section:object-def}

The following is an analog of Definition~\ref{def:fast-task} for the consensus object.
\begin{definition}
  \label{def:ets-obj}
  A protocol is \textbf{$\boldsymbol{e}$-two-step} if for all $E \subseteq \procSet$ of size $e$:
  \begin{enumerate}[noitemsep,topsep=0pt,parsep=0pt]
  \item For every value $v$ and process $p \in \procSet \setminus E$, there exists an $E$-faulty synchronous run in which only $p$ calls $\propose()$, the proposed value is $v$, and the run is two-step for $p$.
    \label{def:ets-obj-one}
  \item For every value $v$ and process $p \in \procSet \setminus E$, there exists an $E$-faulty synchronous run in which all processes in $\procSet \setminus E$ call $\propose(v)$ at the beginning of the first round, and the run is two-step for $p$.
    \label{def:ets-obj-two}
  \end{enumerate}
\end{definition}
As we show in \S\ref{section:lower-2}, item~\ref{def:ets-obj-two} is actually not required for the lower bound, and only serves to make the upper bound stronger.

  \section{Proofs of Lower Bounds}
\label{section:lower-proof}

\subsection{Consensus Task (Theorem~\ref{lower:task}, ``only if'')}
\label{section:lower-1}

Assume there exists an $f$-resilient $e$-two-step protocol $\mathcal{P}$ implementing a consensus task for $n$ processes (Definition~\ref{def:fast-task}).
The fact that $n \ge 2f+1$ follows from well-known results~\cite{dls}.
We now prove the rest by by contradiction.
Thus, assume $n = 2e + f - 1$.
Given a set of processes $P$ and a run $\run$, $\steps{K}{P}{\run}$ denotes the steps taken by $P$ in the $K$-th round of $\run$.
If no such steps were taken, $\steps{K}{P}{\run}$ is empty.
We use the following simple lemma, whose proof we omit.

\begin{lemma}
  Any $E$-faulty synchronous run starting from a configuration $I$ can also start from any configuration $I'$ in which all processes, except those belonging to $E$, propose the same value as in $I$.
  \label{lem:init}
\end{lemma}

\begin{lemma}
  Assume $e \geq 1$.
  Let $k$ be any value such that $0 \leq k \leq \lfloor (f-1)/e \rfloor$.
  Consider any partitioning of $\procSet$ into four sets $E_0$, $E_1$, $F_0$ and $F_1$, such that $|E_0| = |E_1| = e$, $|F_0| = f-1-ke$ and $|F_1| = ke$.
  For any initial configuration $I_k$ in which all processes in $E_0 \cup F_0$ propose $0$ and all processes in $E_1 \cup F_1$ propose $1$,
  there exists an $E_0$-faulty synchronous run of $\mathcal{P}$ starting from $I_k$ which is two-step for a process deciding $0$.
  \label{prop:zero}
\end{lemma}
\begin{proof}
  By induction on $k$. {\em Base case: $k = 0$.}
  Notice that in this case $F_1 = \emptyset$.
  Since $\mathcal{P}$ is an $e$-two-step protocol and since $|E_0| = e$, there exists an $E_0$-faulty synchronous run $\run$ starting from $I_0$ which is two-step for a process $p$.
  In $\run$ the decided value is either $0$ or $1$, as these are the only proposed values.
  We prove by contradiction that it is $0$.
  Assume the converse.
  Since $|E_1| = e$, there exists an $E_1$-faulty synchronous run $\run'$ starting from $I_0$ which is two-step for a process $p' \in E_0 \cup F_0$.
  Notice that in $\run'$ all correct processes (i.e., $E_0 \cup F_0$) propose $0$.
  Thus, in $\run'$ the decided value is $0$.
  Moreover, we can assume without loss of generality that $p' \in F_0$.
  Let
  \begin{equation*}
    \run_1 = \steps{1}{E_1 \cup F_0}{\run}\steps{2}{E_1 \cup F_0}{\run}{\tt decide}(p, 1)\steps{1}{E_0}{\run'}\steps{2}{E_0}{\run'}\mathtt{crash}(F_0\cup\{p\}).
  \end{equation*}
  That is, $\run_1$ is constructed as follows:
  \begin{itemize}[noitemsep,topsep=5pt,parsep=0pt]
  \item Processes in $E_1 \cup F_0$ execute the same first two steps they execute in $\run$, and process $p$ decides $1$.
  \item Processes in $E_0$ execute the first step they execute in $\run'$.
  \item Processes in $E_0$ receive the messages they sent on the previous step, and the messages sent by the processes in $F_0$ on the first step of $\run$ (notice that these messages are identical in $\run$ and $\run'$).
  \item Processes in $E_0$ execute the second step of $\run'$ (notice that this is enabled by the previous item). %
  \item Processes in $F_0 \cup \{p\}$ fail (notice that $|F_0 \cup \{p\}| \leq f$).
  \end{itemize}
  Symmetrically, there exists a valid run $\run_0$ starting from $I_0$, such that
  \begin{equation*}
    \run_0 = \steps{1}{E_0 \cup F_0}{\run'}\steps{2}{E_0 \cup F_0}{\run'}{\tt decide}(p', 0)\steps{1}{E_1}{\run}\steps{2}{E_1}{\run}\mathtt{crash}(F_0\cup\{p\}).
  \end{equation*}
  Notice that the steps taken by the correct processes (i.e., $\procSet \setminus (F_0 \cup \{p\})$) are identical in both runs $\run_1$ and $\run_0$.
  This means that any continuation of one run is also a continuation of the other.
  However, since $\mathcal{P}$ is $f$-resilient, there exists a continuation of $\run_1$ where processes decide $1$, and a continuation of $\run_0$ where they decide $0$ -- a contradiction.

  {\em Induction step.}
  Assume that the lemma holds for $k-1 \geq 0$; we prove it also holds for $k$.
  Since $\mathcal{P}$ is an $e$-two-step protocol and $|E_0| = e$, there exists an $E_0$-faulty synchronous run $\run$ starting from $I_k$ which is two-step for a process $p$.
  In $\run$ the decided value is either $0$ or $1$.
  We prove by contradiction that it is $0$.
  Assume the converse.
  Let $q$ be any process in $F_1$ (notice that $|F_1| = ke \geq e \geq 1$).
  Let $\hat{E_1}$ and $\hat{F_1}$ be two sets such that: if $p \in F_1 \cup F_0$, then $\hat{E_1} = E_1$ and $\hat{F_1} = F_1$; and if $p \in E_1$, then $\hat{E_1} = (E_1 \setminus \{p\}) \cup \{q\}$ and $\hat{F_1} = (F_1 \setminus \{q\}) \cup \{p\}$.
  The sets $\hat{E_1}$, $\hat{F_1}$, $E_0$ and $F_0$ partition $\procSet$.
  They also have the same cardinalities and propose the same values in $I_k$ as sets $E_1$, $F_1$, $E_0$ and $F_0$, respectively.

  Let $I_{k-1}$ be a configuration in which all processes except those in $\hat{E_1}$ propose the same values as in $I_k$, and processes in $\hat{E_1}$ propose $0$ instead of $1$.
  Let $E'_1$ be a subset of $\hat{F_1}$ of $e$ processes.
  This set is well defined because $|\hat{F_1}| = |F_1| = ke \geq e$.
  Moreover, let $F'_1 = \hat{F_1} \setminus E'_1$, $F'_0 = F_0 \cup E_0$ and $E'_0 = \hat{E_1}$.
  By our induction hypothesis, there exists an $E_0'$-faulty synchronous run $\run'$ starting from $I_{k-1}$ which is two-step for a process $p' \in \procSet \setminus E'_0$ that decides $0$.
  By Lemma~\ref{lem:init} the same run can start from $I_k$.
  Similar to the base case, using $\run$ and $\run'$ we construct two valid runs starting from $I_k$:
  \begin{equation*}
    \run_1 = \steps{1}{\hat{E_1} \cup \hat{F_1} \cup F_0}{\run}\steps{2}{\hat{E_1} \cup \hat{F_1} \cup F_0}{\run}{\tt decide}(p, 1)\steps{1}{E_0}{\run'}\steps{2}{E_0}{\run'}\mathtt{crash}(\hat{F_1} \cup F_0\cup\{p'\}),
  \end{equation*}
  \begin{equation*}
    \run_0 = \steps{1}{E_0 \cup \hat{F_1} \cup F_0}{\run'}\steps{2}{E_0 \cup \hat{F_1} \cup F_0}{\run'}{\tt decide}(p', 0)\steps{1}{\hat{E_1}}{\run}\steps{2}{\hat{E_1}}{\run}\mathtt{crash}(\hat{F_1} \cup F_0 \cup\{p'\}).
  \end{equation*}
  There exists a continuation of $\run_1$ where processes decide $1$, and a continuation of $\run_0$ where they decide $0$.
  Since the correct processes (i.e., $(\hat{E_1} \cup E_0) \setminus \{p'\}$) perform the same actions in $\run_1$ and $\run_0$,  we conclude with the same kind of contradiction as in the proof of the base case above.
\end{proof}

We can now conclude the proof of Theorem~\ref{lower:task} using the results above.
Take a partitioning $E_0$, $E_1$, $F_0$, $F_1$ and a configuration $I_k$ satisfying Lemma~\ref{prop:zero} for $k = \lfloor (f-1)/e \rfloor$.
Then there exists an $E_0$-faulty synchronous run $\run$ starting from $I_k$ which is two-step for a process $p \in \procSet \setminus E_0$ that decides $0$.
By Lemma~\ref{lem:init}, $\run$ can start from a configuration $I$, where processes in $E_1 \cup F_0 \cup F_1$ propose the same values as in $I_k$, and those in $E_0$ propose $1$.
Let $E \subseteq \procSet$ be any set of size $e$ such that $F_0 \subseteq E$ (recall that $|F_0| = f-1-ke < e$).
Let $F = \procSet \setminus (E \cup E_0)$.
Notice that $|F| = f-1$.
Since $|E| = e$, there exists an $E$-faulty synchronous run $\run'$ starting from $I$ which is two-step for a process $p' \in F$.
In $\run'$, all correct processes propose $1$ and processes in $E$ are crashed from the beginning.
Hence, $p'$ decides $1$ in $\run'$.
Similar to the proof of Lemma~\ref{prop:zero}, using $\run$ and $\run'$ we construct two valid runs starting from $I$:
\begin{equation*}
  \run_0 = \steps{1}{E \cup F}{\run}\steps{2}{E \cup F}{\run}{\tt decide}(p, 0)\steps{1}{E_0}{\run'}\steps{2}{E_0}{\run'}\mathtt{crash}(F\cup\{p\}),
\end{equation*}
\begin{equation*}
  \run_1 = \steps{1}{E_0 \cup F}{\run'}\steps{2}{E_0 \cup F}{\run'}{\tt decide}(p', 1)\steps{1}{E}{\run}\steps{2}{E}{\run}\mathtt{crash}(F\cup\{p\}).
\end{equation*}
We then obtain the same kind of a contradiction as in the proof of Lemma~\ref{prop:zero}.

  \subsection{Consensus Object (Theorem~\ref{lower:object}, ``only if'')}
\label{section:lower-2}

Assume there exists an $f$-resilient $e$-two-step protocol $\mathcal{P}$ implementing a consensus object for $n$ processes (Definition~\ref{def:ets-obj}).
The fact that $n \ge 2f+1$ follows from well-known results~\cite{dls}.
We now prove the rest by by contradiction.
Thus, assume $n = 2e + f - 2$.
Consider two distinct processes $p$ and $q$.
Let $E_0$ and $E_1$ be two quorums such that
$p \in E_0$, $q \in E_1$, and
$\cardinalOf{E_0} = \cardinalOf{E_1} = n-e$.
Let $F = E_0 \inter E_1$, $E_0^{*}=E_0 \setminus (E_1 \union \{p\})$, and $E_1^{*}=E_1 \setminus (E_0 \union \{q\})$.
Notice that since $n = 2e + f -2$, $|F| = n-2e = f-2$.

Since $\mathcal{P}$ is $e$-two-step, there exists a run in which all the processes in $\procSet \setminus E_0$ are initially crashed, process $p$ is the only process to propose, and it decides at time $2\Delta$.
Assume that in this run, $p$ proposes and decides the value $0$.
Consider the prefix $\run_0$ of this run until process $p$ decides. %
Symmetrically, there exists a run $\run_1$ where all processes in $\procSet \setminus E_1$ are initially crashed and $q$ decides $1$ after two message delays.
Let
\begin{equation*}
  \run = \steps{1}{F \union E_0^* \union \{p\}}{\run_0}\steps{1}{E_1^* \union \{q\}}{\run_1}\crash(F \union \{p,q\})\steps{2}{E_0^*}{\run_0}\steps{2}{E_1^*}{\run_1}.
\end{equation*}
That is, $\run$ is constructed as follows:
\begin{itemize}
\item Processes in $E_0$ execute the first step they execute in $\run_0$.
\item Processes in $E_1^* \union \{q\}$ execute the first step of $\run_1$.
\item Processes in $F \union \{p,q\}$ fail (notice that $F \union \{p,q\} = f$).
\item Processes in $E_0^* \cup E_1^*$ receive the messages they sent on the previous step, and the messages sent by the processes in $F$ on the first step of $\run_0$ (notice that these messages are identical in $\run_0$ and $\run_1$).
\item Processes in $E_0^*$ execute the second step of $\run_0$, and processes in $E_1^*$ execute the second step of $\run_1$ (this is enabled by the previous item).
\end{itemize}
The processes outside $\{p,q\}$ do not propose neither in $\run_0$, nor in $\run_1$.
Hence, $\steps{1}{F}{\run_0} = \steps{1}{F}{\run_1}$.
As a consequence, $E_0$ and $E_1^*$ cannot distinguish $\run$ from $\run_0$ and $\run_1$, respectively.
Furthermore, since every message received in $\run$ was sent previously, $\run$ is well-formed.
It follows that $\run$ is a partial run of $\mathcal{P}$.

Because $f$ processes failed in $\run$ and the protocol $\mathcal{P}$ is $f$-resilient, there exists a continuation $\hat{\run}$ of $\run$ during which a decision is taken.
Without loss of generality, assume that $1$ is decided in $\hat{\run}$.
Let $r \in E_0^* \union E_1^*$ be the process deciding first in $\hat{\run}$.
The continuation $\hat{\run}$ is of the form $\hat{\run}=\run\lambda$, with $\decide(r, 1) \in \lambda$.

Assume the following partial asynchronous run:
\begin{equation*}
\run' =
\steps{1}{E_0}{\run_0}
\steps{1}{E_1^* \union \{q\}}{\run_1} %
\steps{2}{E_0}{\run_0}
\steps{2}{E_1^*}{\run_1}
\crash(F \cup \{q\}) %
\decide(p, 1)
\crash(\{p\})
\lambda.
\end{equation*}
In this run, process $p$ fails after deciding in the third round.
Processes $F$ and $\{q\}$ crash at the end of the second round.

The run $\run'$ is well-formed.
It is not distinguishable from $\run_0$ to processes $F \union \{p\}$.
To the other processes, run $\run'$ is not distinguishable from $\hat{\run}$.
Hence, it is a partial run of $\mathcal{P}$.
In $\run'$, process $r$ decides $1$ while process $p$ decides $0$;
contradiction.

  \section{Proofs of Upper Bounds}
\label{section:upper-proof}

\subsection{Leader Election}
\label{section:full-new-ballot}

To ensure the {\em Termination} property of consensus, in addition to the logic described in \S\ref{sec:upper}, our protocol relies on an $\Omega$ leader election service, which can be implemented under partial synchrony in a standard way~\cite{omega}.
At each process $p_i$, $\Omega$ outputs a process that $p_i$ currently believes to be correct.
The $\Omega$ service guarantees that eventually all correct processes agree on the same correct leader.
The complete version of the handler at line~\ref{alg:timeout} uses $\Omega$ as follows:
\vspace{3pt}
\begin{center}
\begin{minipage}{0.95\linewidth}
  \removelatexerror
  \scalebox{0.9}{
  \begin{algorithm}[H]
    \setstretch{0.9}
    \renewcommand{\;}{\\}
    \DontPrintSemicolon

    \SubAlgo{{\bf when the timer $\timer$ expires}}{
      $\starttimer(\timer, 5\Delta)$\;
      \label{alg:reset-timer}
      \If{$\Omega = p_i$}{
        \label{alg:omega-if}
        \Let $b = (\mbox{\rm a ballot $> \bal$ such that}$ $i \equiv b\ ({\tt mod}\ n))$\;
        \send $\MoneA(b)$ \To{} $\Pi$\;
      }
    }
  \end{algorithm}
}
\end{minipage}
\end{center}
\vspace{3pt}

Thus, a process $p_i$ initiates a new ballot only if $\Omega$ identifies $p_i$ as the leader (line~\ref{alg:omega-if}).
This ensures that eventually only one process tries to take over, preventing situations where multiple processes interfere with each other's attempts to reach agreement on the slow path.

Additionally, the code above makes explicit use of a timer $\timer$.
Initially it is set to $2\Delta$, giving just enough time for the processes to reach agreement on the fast path.
After this, the timer is reset with a delay of $5\Delta$ (line~\ref{alg:reset-timer}).
After $\mathsf{GST}$, this duration is sufficient for all processes to decide on the slow path.

\subsection{Consensus Task (Theorem~\ref{lower:task}, ``if'')}

We already argued in \S\ref{sec:upper} that the protocol in Figure~\ref{fig:alg} is $e$-two-step according to Definition~\ref{def:fast-task}.
We now prove that it satisfies the consensus specification.
{\em Validity} trivially follows from the structure of the protocol.
{\em Termination} is easily ensured by the leader election service (\S\ref{section:full-new-ballot}).
{\em Agreement} follows from Lemma~\ref{lem:fast-path-rec} and the following lemma.
\begin{lemma}
  \label{lem:slow-path}
  Given $b > 0$, assume that at least $n-f$ processes have received $\MtwoA(b, v)$ and replied with $\MtwoB(b, v)$.
  Then for any $b' \geq b$ and any message $\MtwoA(b', v')$ sent we have $v' = v$.
\end{lemma}
\begin{proof}
  We prove the following statement by induction on $b'$: for all $b'$ and $b$ such that $b' \geq b > 0$, if at least $n-f$ processes have received $\MtwoA(b, v)$ and replied with $\MtwoB(b, v)$, then for any message $\MtwoA(b', v')$ sent we have $v' = v$.
  Assume that this holds for all values $< b'$; we show that it also holds for $b'$.

  Consider an arbitrary $b > 0$ and assume that at least $n-f$ processes have received $\MtwoA(b, v)$ and replied with $\MtwoB(b, v)$.
  Assume also that $\MtwoA(b', v')$ has been sent, which must have happened when a process $p$ by executed the handler at line~\ref{alg:receive1b}.
  Assume that $p$ received $\MoneB(b', \lvbal_q, v_q, \lproposer_q, \ldecided_q)$ from every process $q$ in a quorum $Q'$ and let $\bmax = {\tt max}\{\lvbal_q \mid q \in Q'\}$ (line~\ref{alg:let-bmax}).

  Since $n \geq 2f+1$ and $|Q| = |Q'| = n-f$, we have $Q \cap Q' \neq \emptyset$.
  Let $q \in Q \cap Q'$ be an arbitrary process in the intersection.
  The process $q$ sends both messages $\MoneB(b', \lvbal_q, v_q, \lproposer_q, \ldecided_q)$ and $\MtwoB(b, v)$.
  After sending $\MoneB(b', \lvbal_q, v_q, \lproposer_q, \ldecided_q)$ it has $\bal \geq b'$ (line~\ref{alg:one-a-set-bal}).
  Since $b' > b$, by line~\ref{alg:two-a-cond}, process $q$ sends the $\MtwoB$ message before $\MoneB$.
  Then by line~\ref{alg:two-a-set-vars} at the moment it sends $\MoneB(b', \lvbal_q, v_q, \lproposer_q,\ldecided_q)$, it has $\vbal \geq b$.
  This proves that $\bmax \geq b$.

  Note that since $\bmax \geq b > 0$, the value $v'$ is computed either at line~\ref{alg:set-learn} or line~\ref{alg:set-slow-ballot}.
  \begin{itemize}
  \item Assume fist that $v'$ is computed at line~\ref{alg:set-learn}.
    By line~\ref{alg:cond-learn}, there exists a process $q' \in Q'$, such that $v' = \ldecided_{q'}$.
    We now consider three cases: $\lvbal_{q'} > b$, $\lvbal_{q'} = b$, and $\lvbal_{q'} < b$.
    If $\lvbal_{q'} > b$ then by induction hypothesis $v' = \ldecided_{q'} = v$, as required.
    If $\lvbal_{q'} = b$ then process $q'$ must have received the $\MDecide(\ldecided_{q'})$ message at ballot $b$.
    This message is sent by the same process that previously sent the $\MtwoA(b, v)$ message.
    Hence, $v' = \ldecided_{q'} = v$, as required.
    Finally, assume that $\lvbal_{q'} < b$.
    If $\lvbal_{q'} = 0$ then by Lemma~\ref{lem:fast-path-rec}, $v' = \ldecided_{q'} = v$.
    If $\lvbal_{q'} > 0$ then there must exist a quorum of processes sending $\MtwoB(\lvbal_{q'}, v_{q'})$ messages.
    Hence, by induction hypothesis $v' = v_{q'} = v$, as required.
  \item Assume now that $v'$ is computed at line~\ref{alg:set-slow-ballot}.
    Take any process $q' \in Q$ that has $\lvbal_{q'} = \bmax$.
    We consider separately the cases when $\bmax > b$ and $\bmax = b$.
    If $\bmax > b$ then by induction hypothesis $v_{q'} = v$.
    If $\bmax = b$ then, since $b > 0$, process $q'$ must have accepted the value received in the $\MtwoA(b, v)$ message.
    Thus, in both cases $v_{q'} = v$, and hence, $v' = v_{q'} = v$, as required.
  \end{itemize}
\end{proof}

\balance

\subsection{Consensus Object (Theorem~\ref{lower:object}, ``if'')}

We consider the protocol in Figure~\ref{fig:alg} with {\color{purple} red lines} included.
The proof that this protocol satisfies the consensus specification is analogous to that of Theorem~\ref{lower:task}(``if''), but relies on the following modified version of Lemma~\ref{lem:fast-path-rec}.

\begin{lemma}
  Assume $n \geq 2e + f - 1$.
  If a value $v$ is decided via the fast path at ballot $0$, then lines~\ref{alg:cond-1}--\ref{alg:set-val-initial} always select $v$ as $\lval$. 
  \label{lem:fast-path-rec-obj}
\end{lemma}
\begin{proof}
  Since $v$ takes the fast path, there are at least $n-e$ processes that vote for $v$ at ballot $0$ (implicitly including the proposer of $v$).
  The set of these processes intersects with $Q$ (line~\ref{alg:receive1b}) in $\geq n-e-f$ processes, thus enabling the condition at either line~\ref{alg:cond-1} or line~\ref{alg:cond-2}.

  Consider first the case of line~\ref{alg:cond-1}, so that $>n-e-f$ processes in $Q$ voted for $v$.
  Then the number of processes in $Q$ that could have voted for another value is
  $$
  \leq (n-f)-(n-e-f)-1=e-1=(2e+f-1)-e-f\leq n-e-f.
  $$
  Hence, $v$ is the only value that can satisfy the condition at line~\ref{alg:cond-1}, and $\lval$ must be assigned to $v$ at line~\ref{alg:set-lval-cond-1}.

  Assume now that the condition at line~\ref{alg:cond-2} holds, so that $n-e-f$ processes in $Q$ voted for $v$.
  Since $v$ went onto the fast path, $\ge n-e$ processes voted for $v$ overall.
  Thus, all $f$ processes outside of $Q$ must have voted for $v$.
  Then none of them could have sent a $\MPropose$ message for a value different from $v$ (lines~\ref{alg:startup-precond} and \ref{alg:propose-precond}).
  Since the condition at line~\ref{alg:cond-2} holds only for the values proposed by the processes outside of $Q$ (line~\ref{alg:q-prime}), value $v$ is the only value that can satisfy this condition, and $\lval$ must be assigned to $v$ at line~\ref{alg:set-lval}.
\end{proof}

We now prove that the protocol in Figure~\ref{fig:alg} with {\color{purple} red lines} included is $e$-two-step according to Definition~\ref{def:ets-obj}.
Consider any value $v$, a set $E \subseteq \Pi$ of size $e$, and a process $p \notin E$.
We consider each of the conditions in Definition~\ref{def:ets-obj} separately:
\begin{enumerate}
\item In any $E$-faulty synchronous run where only $p$ calls $\propose()$ and the proposed value is $v$, all other correct processes will be able to satisfy the precondition at line~\ref{alg:propose-precond} and respond to $p$ with a $\MtwoB$ message. 
  Process $p$ will then satisfy the condition at line~\ref{alg:two-b-precond} and decide $v$ by $2\Delta$.
\item If all correct processes call $\propose(v)$ at the start of the first round, then the $\MPropose$ message from any correct process $p$ can be the first one received by all others, ensuring that $p$ decides in $2\Delta$.
\end{enumerate}

\fi

\end{document}